\documentclass[11pt, a4paper]{article}

\pdfoutput=1

\usepackage{latexsym}
\usepackage[centertags]{amsmath}
\usepackage{amsfonts}
\usepackage{newlfont}
\usepackage{graphicx}
\usepackage{float}
\usepackage{hhline}
\usepackage{algorithm}
\usepackage{algpseudocode}
\usepackage{fullpage}
\usepackage{verbatim}
\usepackage{multicol}
\usepackage{lipsum}
\usepackage{enumitem}
\usepackage{subcaption}
\usepackage{caption}
\usepackage{amsmath}
\usepackage{amssymb}
\usepackage{amsthm}
\usepackage{mdframed}
\usepackage{collectbox}
\usepackage{framed}
\usepackage{pbox}
\usepackage{algorithmicx}

\def\denseformat{
\setlength{\textheight}{9.65in} \setlength{\textwidth}{6.9in}
\setlength{\oddsidemargin}{-0.3in} \setlength{\headsep}{10pt}
\setlength{\topmargin}{-0.33in} \setlength{\columnsep}{0.375in}
}
\denseformat

\algdef{SE}[RECEIVING]{Receiving}{EndReceiving}[1]{\textbf{upon
receiving}\ #1\ \algorithmicdo}{\algorithmicend\ \textbf{}}%
\algtext*{EndEvent}

\algdef{SE}[UPON]{Upon}{EndUpon}[1]{\textbf{upon
}\ #1\ \algorithmicdo}{\algorithmicend\ \textbf{}}%
\algtext*{EndEvent}

\algdef{SE}[OPERATION]{Operation}{EndOperation}[1]{\textbf{operation}\
#1\ }{\algorithmicend\ \textbf{}}%
\algtext*{EndEvent}

\restylefloat{table}
\restylefloat{verbatim}

\newenvironment{thmclone}[1]{\noindent
\textbf{Theorem~\ref{#1}.}\em}{\par}

\newmdenv[linewidth=0pt]{myframe}

\theoremstyle{definition}
\newtheorem{definition}{Definition}

\theoremstyle{plain}

\newtheorem{observation}{Observation}
\newtheorem{corollary}{Corollary}
\newtheorem{theorem}{Theorem}  
\newtheorem{lemma}{Lemma}

\newtheorem{invariant}{Invariant}

\theoremstyle{remark}



\begin{document}

\title{Space Bounds for Reliable Storage:\\ Fundamental Limits
of Coding}
\author{%
  {Alexander Spiegelman}\\EE Department\\Technion, Haifa,
  Israel\\sashas@tx.technion.ac.il\\+972547553558
  \and
  {Yuval Cassuto}\\EE Department\\Technion, Haifa,
  Israel\\ycassuto@ee.technion.ac.il\\
  \and
  {Gregory Chockler}\\CS Department\\Royal Holloway, London,
  UK\\gregory.chockler@rhul.ac.uk\\
  \and
  {Idit Keidar}\\EE Department\\Technion, Haifa,
  Israel\\idish@ee.technion.ac.il\\
}

\date{}
\maketitle

\begin{abstract}

We study the inherent space requirements of shared storage
algorithms in asynchronous fault-prone systems. 
Previous works use codes to achieve a better storage cost than
the well-known replication approach. 
However, a closer look reveals that they
incur extra costs somewhere else: Some use unbounded storage in
communication links, while others assume bounded concurrency or
synchronous periods. We prove here that this is inherent, and
indeed, if there is no bound on the concurrency level, then the
storage cost of any reliable storage algorithm is at least $f+1$
times the data size, where $f$ is the number of tolerated
failures. 
We further present a technique for combining
erasure-codes with full replication so as to obtain the best of
both. 
We present a storage algorithm whose storage cost is close
to the lower bound in the worst case, and adapts to the
concurrency level.

\end{abstract}

\setcounter{page}{0}
\thispagestyle{empty}
 
\newpage

\section{Introduction}
\label{introduction}

We reason about the storage space required for
emulating reliable shared storage over fault-prone nodes.
The traditional approach to building 
such storage stores full
replicas of the data in each node~\cite{ABD}. 
This approach entails a fixed storage cost
equal to the size of the data times the
number of nodes, regardless of the level of
concurrency.

Recently, there is an active area of
research
of employing codes, and in particular erasure codes, in
distributed algorithms with the goal of reducing the storage
cost~\cite{codingAguilera2005using,codingCachin2006optimal,codingGoodson2004efficient,codingLynchCadambe2014coded,codingMultiVersion,codingRashid}.
But when we look at these works closely, we find that in all
asynchronous solutions, extra costs are hidden somewhere.
Some keep an unbounded number of
versions~\cite{codingGoodson2004efficient},
or as many as the allowed level of
concurrency~\cite{codingLynchCadambe2014coded}. Others keep
unbounded information in
channels~\cite{codingRashid,codingCachin2006optimal}.
While others
assume periods of synchrony~\cite{codingAguilera2005using} or
allow returning obsolete values \cite{codingMultiVersion}.

To provide intuition about erasure-coded reliable storage
algorithms, we give in Section \ref{sec:WFsafe} a simple
space-efficient solution that only guarantees safe semantics
\cite{lamportRegular}, which are too weak to be of practical
use. We use this example to illustrate the challenges that have
led algorithms that provide stronger semantics to store many
versions of the coded data.

Then, in Section \ref{sec:impossibility}, we prove  that this
is inherent: any lock-free algorithm that
simulates reliable storage in an asynchronous system where $f$
storage nodes can fail must sometimes store $f+1$ full replicas
of written data, or its storage cost can grow without bound.
Specifically, our bound applies to any fault-tolerant
implementation of a \emph{multi-writer multi-reader (MWMR)
register} that satisfies at least \emph{weak regularity}, a
safety notion weaker than linearizability.

We prove our result for the fault-prone shared memory 
model~\cite{faultyMemoryAfek1993benign,faultyMemorybyzantineDiscPaxos,faultyMemoryJayanti1998}
in order to avoid reasoning explicitly about channels. The same
bound applies to message passing systems if we limit the
capacity of communication channels. 
For the sake of our proof, we define a
specific adversary behavior, which makes the proof fairly
compact.

Understanding the inherent
storage cost limitation that stems from our lower bound, and
in particular, the fact that, under high concurrency, nodes have
to keep full replicas, leads us to develop an adaptive approach
that combines the advantages of full replication and coding.  
We present in Section \ref{sec:WRAlgorithms}
an algorithm that simulates
an FW-Terminating~\cite{faultyMemorybyzantineDiscPaxos} strongly
regular~\cite{regularMWMR} MWMR register, whose storage
requirement is close to the storage limitation in the worst
case, and uses less storage in runs with low concurrency. The
algorithm does not assume any a priori bound on concurrency;
rather, it uses erasure codes when concurrency is low and
switches to replication when it is high.

Finally, we believe that our work is only a first effort to
combine erasure coding with replication in order to achieve
adaptive storage costs. We conclude in Section
\ref{sec:discussion} with some thoughts about directions for
future work.

\vspace*{-.3cm}
\section{Preliminaries}
\label{sec:model}

\vspace*{-.2cm}
\subsection{Model}
\label{sub:model}

We consider an asynchronous fault-prone shared memory
system~\cite{faultyMemoryAfek1993benign,faultyMemorybyzantineDiscPaxos,faultyMemoryJayanti1998}
consisting of set $N=\{bo_i,\ldots,bo_n\}$ of base objects
supporting arbitrary atomic \emph{read-modify-write} (RMW)
access by clients from some finite set $\Pi$. Any $f$ base
objects and any number of clients may fail by crashing, for some
predefined $f < n/2$. We study algorithms that emulate a shared
object to a set of clients. 

Clients interact with the emulated object
via high-level \emph{operations}. To distinguish the high-level
emulated operations from low-level base object access, we refer
to the latter as \emph{RMWs}. We say that RMWs are
\emph{triggered} and \emph{respond}, whereas operations are
\emph{invoked} and \emph{return}. A (high-level) operation
consists of a series of trigger and respond \emph{actions} on
base objects, starting with the operation's invocation and
ending with its return. 
In the course of an operation, a client \emph{triggers}
RMWs separately on each $bo_i \in N$ and receives
\emph{responses} in return.
We model the state of each $bo_i \in N$ as
changing, according to the RMW triggered on
it, at some point after the time when the RMW is triggered
but no later than the time when the matching response occurs.

An \emph{algorithm}  defines the behavior of clients as
deterministic state machines, where state transitions are
associated with actions such as RMW trigger/response.
A \emph{configuration} is a mapping to states from system
components, i.e., clients and base objects. An
\emph{initial configuration} is one where all components are
in their initial states.

A \emph{run} of algorithm $A$ is a (finite or infinite) 
alternating sequence of configurations and actions, beginning
with some initial configuration, such that configuration
transitions occur according to $A$.
We use the notion of
time $t$ during a run $r$ to refer to the configuration
incurred after the $t$\textsuperscript{th} action in $r$.
A \emph{run fragment} is a
contiguous subsequence of a run.

We say that a base object or client is \emph{faulty} in a run
$r$ if it fails any time in $r$, and otherwise, it is
\emph{correct}.
A run is \emph{fair} if (1) for every RMW  triggered by
a correct client on a correct base object, there is eventually a
matching response, (2) every correct client gets infinitely many
opportunities to trigger RMWs.
We again use different terminology to distinguish
incomplete invocations to the high-level service from incomplete
RMWs triggered on base objects and refer to the former as
\emph{outstanding} operations and to the latter as
\emph{pending} RMWs.

Operation $op_i$ \emph{precedes}
operation $op_j$ in a run $r$, denoted $op_i  \prec_r 
op_j$, if $op_i$'s response occurs before $op_j$'s invoke in
$r$. Operations $op_i$ and $op_j$ are
\emph{concurrent} in a run $r$, if neither precedes the other. A
run with no concurrent operations is \emph{sequential}.

\vspace*{-.3cm}
\subsection{Storage service definitions}
\label{sub:problen}

We study emulations of an \emph{MWMR register}, which
stores a value $v$ from a domain $\mathbb{V}$, and offers an
interface for invoking \textit{read} and \textit{write}
operations.
Initially, the register holds some distinguished initial value
$v_0 \in \mathbb{V}$.
The sequential specification for this service is as
follows: A read returns the latest written value, or $v_0$ if
none was written.

The storage resources consumed by the MWMR
register emulations discussed herein are measured in units of
{\em bits}. For constructive algorithmic results, bits are
stored in base objects following writes triggered by clients,
and correctness lies upon the existence of a decoding algorithm
that can recover $v \in \mathbb{V}$ from the bits available to
the reader. The common examples for such decoding algorithms are
1) the trivial decoder mapping $D=\log_2|\mathbb{V}|$ bits to
the value $v$ using the standard binary representation, as in
the case of replication; and 2) an erasure-code decoder mapping
a set of $D$ or more code bits to $v$. For the impossibility
proof we use a fundamental information theoretic argument that
any representation, either coded or unncoded, cannot guarantee
to recover $v$ precisely from fewer than $D=\log_2|\mathbb{V}|$
bits. This argument excludes common storage-reduction techniques
like compression and de-duplication, which only work in
probabilistic setups and with assumptions on the written data.

We now proceed to detail the properties
describing the MWMR register.\newline

 \textbf{Liveness}~~~
There is a range of possible liveness conditions, which need to
be satisfied in fair runs of a storage algorithm.
A \emph{wait-free} object is one
that guarantees that every correct client's operation
completes, regardless of the actions of other clients.
A \emph{lock-free} object 
guarantees progress: if at some point in a run there is an
outstanding operation of a correct client, then
\emph{some} operation eventually completes.
An \emph{FW-terminating}~\cite{faultyMemorybyzantineDiscPaxos}
register is one that has wait-free \emph{write} operations, and
in addition, if there are finitely
many \emph{write} invocations in a run, then every
\emph{read} operation completes.

\textbf{Safety}~~~
Two runs are \emph{equivalent} if every client performs the same
sequence of operations in both, where operations that are
outstanding in one can either be included in or excluded from
the other.
A linearization of a run $r$ is an equivalent sequential
execution that satisfies $r$'s operation precedence relation and
the object's sequential specification.
A \emph{write} $w$ in a run $r$ is \emph{relevant} to a
\emph{read} $rd$ in $r$~\cite{regularMWMR} if $rd \not\prec_r
w$; \emph{rel-writes}$(r,rd)$ is the set of all \emph{writes} in
$r$ that are relevant to $rd$.

Following Lamport~\cite{lamportRegular}, we
consider a hierarchy of safety notions.
Lamport~\cite{lamportRegular} defines \emph{regular} and 
\emph{safe} single-writer registers. Shao et
al.~\cite{regularMWMR} extend Lamport's notion of regularity to
MWMR registers, and give four possible definitions.
Here we use two of them. The first is the weakest definition,
and we use it in our lower bound proof. The second, which we use
for our algorithm, is the strongest definition that is satisfied
by ABD~\cite{ABD} in case readers do not change the storage (no
\emph{write-back}):
A MWMR register is \emph{weakly regular}, (called
\emph{MWRegWeak} in~\cite{regularMWMR}), if for every run $r$
and \emph{read} $rd$ that returns in $r$, there exists a
linearization $L_{rd}$ of the subsequence of $r$ consisting of the write
operations in $r$ and $rd$.
A MWMR register is \emph{strongly regular}, (called
\emph{MWRegWO} in~\cite{regularMWMR}), if it satisfies weak
regularity and the following condition:
For all
\emph{reads} $rd_1$ and $rd_2$ that return in $r$, for all
writes $w_1$ and $w_2$ in $\emph{rel-writes}(r, rd_1) \cap
\emph{rel-writes}(r, rd_2)$, it holds that $w_1 \prec_{L_{rd_1}}
w_2$ if and only if $w_1 \prec_{L_{rd_2}} w_2$.

We extend the safe register definition and
say that a MWMR register is \emph{strongly safe} if there exists
a linearization $\sigma_w$ of the subsequence of $r$ consisting
of the \emph{write} operations in $r$, and for every
\emph{read} operation $rd$ that has no concurrent \emph{writes} 
in $r$, it is possible to add $rd$ at some point in $\sigma_w$
so as to obtain a linearization of the subsequence of $r$
consisting of the write operations in $r$ and $rd$.

\vspace*{-.3cm}
\subsection{Erasure codes}
\label{sub:erasureCode}

A $k$-of-$n$ erasure code takes a value
from  domain $\mathbb{V}$ and produces a set $S$ of $n$
\emph{pieces} from some domain $\mathbb{E}$ s.t.\ the value can
be restored from any subset of $S$ that contains no less than $k$
different pieces. We assume that
the size of each piece is $D/k$, and two
functions \emph{encode} and \emph{decode} are given:
\emph{encode} gets a value $v \in
\mathbb{V}$ and returns a set of $n$ ordered elements
$W = \{\langle v_1,1 \rangle,\ldots,\langle v_n,n \rangle\}$,
where $v_1,\ldots,v_n \in \mathbb{E}$, and \emph{decode} gets a
set $W' \subset \mathbb{E} \times \mathbb{N}$ and returns $v'
\in \mathbb{V}$ s.t.\ if $|W'| \geq k$ and $ W' \subseteq W$,
then $v=v'$.
In this paper we use $k = n-2f$. Note that when $k=1$, we get
full replication.

\section{A Simple Algorithm}
\label{sec:WFsafe}

In order to develop intuition for the structure and limitations
of distributed storage algorithms, we present in Section
\ref{subsec:safe} a simple storage-efficient algorithm that
ensures \emph{safe} semantics, but not \emph{regularity}. 
Although this algorithm has no practical use, it shows that
the impossibility result of Section \ref{sec:impossibility} does
not apply to a weaker safety property.
In Section \ref{subsec:noRegularity}, we then illustrate 
how this simple algorithm can be extended to ensure regularity
using unbounded storage (similarly to some previous works),
as proven to be inherent by our main result in the next section.

\subsection{Safe and wait-free algorithm}
\label{subsec:safe}

This algorithm simulates a wait-free and strongly
safe MWMR register using erasure codes.
It stores exactly $n$ pieces of the data, one in each
base object.
The algorithm's definitions are presented in Algorithm
\ref{alg:WFdefinitions}, and the algorithm of client $c_j$ can
be found in Algorithm \ref{alg:WFoperations}. 

We define $Timestamps$ to be the set of timestamps $\langle
num,c \rangle$, s.t. $num \in \mathbb{N}$ and $c \in \Pi$,
ordered lexicographically. 
We define $Pieces$ to be the set of pairs consisting of
an element from $\mathbb{E}$ (possible
outputs of the \emph{encode} function) and a number, and
$Chunks = Pieces \times Timestamps$. Each base
object $bo_i$ stores exactly one value from \emph{Chunks},
initially $\langle \langle v_{0_i}, i \rangle, \langle 0, 0 \rangle \rangle$, where
$v_{0_i}$ is the $i^{th}$ piece of $v_0$.

Since memory is
fault-prone, actions are triggered in parallel on all base
objects. This parallelism is denoted using \textbf{$||$for} in
the code. Operations then wait for $n-f$ base objects to
respond.
Recall that $n=2f+k$, so every two sets of $n-f$ base objects
have at least $k$ pieces in common.
Thus, if a write completes after storing pieces on $n-f$ base
objects, a subsequent read accessing any $n-f$ base objects finds
$k$ pieces of the written value (as needed for restoring the
value), provided that they are not over-written by later writes.

A \emph{write$(v)$} operation (lines
\ref{line:SafeWbegin}--\ref{line:SafeWend}) first produces $n$
pieces from $v$ using \emph{encode}, then reads from $n-f$ base objects to obtain a new
timestamp, and finally, tries to store every piece together with
the timestamp at a different base object. For every base object
$bo$, $c_j$ triggers the \emph{update} RMW function, which
overwrites $bo$ only if $c_j$'s timestamp is bigger than the
timestamp stored in $bo$.

A \emph{read} (lines
\ref{line:SafeRbegin}--\ref{line:SafeRend}) reads
the values stored in $n-f$ base objects, and then tries to restore valid data as follows.
If $c_j$ reads at least $k$ values with the same timestamp, it
uses the \emph{decode} function, and returns the restored value.
Otherwise, it returns $v_0$. 
The latter occurs
only if there are outstanding \emph{writes}, that had updated
fewer than $n-f$ base objects before the reader has accessed
them. 
Therefore, these \emph{writes} are concurrent with $c_j$'s
\emph{read}, and by the safety property, any value can be
returned in this case.
The algorithm's correctness is formally proven in Appendix
\ref{AppSub:WF}.


\begin{algorithm}[H]
 \caption{Definitions.}
 \label{alg:WFdefinitions}
\begin{algorithmic}[1]
\small

\State $TimeStamps = \mathbb{N} \times \Pi  $, with selectors
$num$ and $c$, ordered lexicographically. 
\State $Pieces = (\mathbb{E} \times \mathbb{N})$
\State $Chunks =Pieces \times TimeStamps$, with selectors
$val,ts$ \State \emph{encode} $: \mathbb{V} \rightarrow
2^{\mathbb{E} \times \{1,2,\ldots,n\}},$
\emph{decode} $: 2^{\mathbb{E} \times \{1,2,\ldots,n\}}
\rightarrow \mathbb{V}$
\State \hspace*{0.5cm} s.t.\ $\forall v \in \mathbb{V}$,
$\emph{encode}(v)=\{\langle *,1 \rangle,\ldots,\langle *,n
\rangle \}  \wedge$
\State \hspace*{0.5cm} $\forall W \in 2^{\mathbb{E} \times
\mathbb{N}}$, 
if $W \subseteq \emph{encode}(v) \wedge
|W| \geq k$, then \emph{decode}$(W)=v$


\end{algorithmic}
\end{algorithm}


\vspace*{-.5cm}

\begin{algorithm}[H]
 \caption{Safe register emulation. Algorithm for client $c_j$.}
 \label{alg:WFoperations}
 \begin{multicols}{2}
 
\begin{algorithmic}[1]
\small

\Operation{$write(v)$}{}
\label{line:SafeWbegin}

\State $W\leftarrow \emph{encode}(v)$   

\State $R \leftarrow \emph{readValue}()$
\State $ts \leftarrow \langle max( \{ts | \langle \langle
ts,* \rangle ,* \rangle \in R \}) + 1, j \rangle$

\State \textbf{$||$ for all} $\langle v,i \rangle \in \emph{W}$ 

\State \hspace*{0.2cm}
$update(bo_i,\langle v,i \rangle ,ts)$
\Comment{trigger RMW on $bo_i$}

\State  \textbf{wait} for $n-f$ responses
\label{line:WFwait1}
\State \textbf{return} ``ok''

\EndOperation
\label{line:SafeWend}

\State \textbf{update}$(bo,w,ts)$
$\triangleq$

\State \hspace*{0.4cm} \textbf{if} $ts > bo.ts$ 
\State \hspace*{0.8cm} $bo \leftarrow \langle w, ts \rangle$ 

\Statex

\Statex
\Operation{$read()$}{}
\label{line:SafeRbegin}

\State $R \leftarrow \emph{readValue}()$
\State \textbf{if}  $\exists ts$ s.t.\ $|\{
v  \mid \langle ts,v \rangle \in R \} | \geq k$ 

\State \hspace*{0.4cm} $ts' \leftarrow ts$ s.t.\ $|\{
v  \mid \langle ts,v \rangle \in R \} | \geq k$

\State \hspace*{0.4cm} \textbf{return} $\emph{decode} ( \{v \mid
\langle ts',v \rangle \in R \})$ 

\State \textbf{return} $v_0$

\EndOperation
\label{line:SafeRend}

\Procedure{$readValue()$}{}

\State $\emph{R} \leftarrow \{\}$
\State \textbf{$||$ for} i=1 to n
\State \hspace*{0.4cm} $\emph{R}= \emph{R} \cup read(bo_i)$ 
\State \textbf{wait} until $|\emph{R}| \geq n-f$ 
\label{line:WFwait2} 
\State \textbf{return} \emph{R}

\EndProcedure

\end{algorithmic}
\end{multicols} 
\end{algorithm}

\vspace*{-.6cm}
\subsection{Achieving regularity with unbounded storage}
\label{subsec:noRegularity}

We now give intuition why extending this approach to satisfy
regularity requires unbounded
storage.
Note that a read from a regular register must return a valid
value even if it has concurrent writes, and that a
write may remain outstanding indefinitely in case the writer
fails.

Consider a system with $n=4$, $f=1$, $k=2$, where $b_1$ is
faulty and clients $c_1$ and $c_2$ invoking $write(v_1)$ and
$write(v_2)$ respectively, as illustrated in Figure
\ref{fig:example_a}.

Since base objects may fail, clients $c_1$ and
$c_2$ try to store their pieces in all the base objects in
parallel (as in Algorithm \ref{alg:WFoperations}). 
Assume that
$c_1$'s first RMW on $b_2$ and
$c_2$'s RMW on $b_3$ take effect.
If these RMWs would overwrite
the pieces in $b_1$ and $b_2$, and $c_1$ and $c_2$ would then
immediately fail, the storage will remain with no restorable
value.
In this case, no later read can return a value satisfying
regularity (note that since the two outstanding writes are
concurrent with any future read, a safe register
may return an arbitrary value).
Therefore, $c_1$ and $c_2$ cannot overwrite the existed value in
the base objects.

Consider next a client $c_3$ attempting to write $v_3$ as in
Figure \ref{fig:example_b}. Even if $c_3$ reads
the base objects, it cannot learn of any complete write. 
Moreover, when its RMW takes effect on $b_4$, it cannot
distinguish between a scenario in which $c_2$ and $c_3$ have
failed (thus, their pieces can be overwritten), and the scenario
in which one of $c_2$ and $c_3$ is slow and will eventually 
be the only client to complete a writes (in which case
overwriting its value may leave the storage with no
restorable value). Thus, $c_3$ cannot overwrite any piece.

We can repeat this process by allowing an unbounded
number of clients to invoke writes and store exactly one piece
each, without allowing any piece to be overwritten.
While this example only shows that a direct extension of
Algorithm \ref{alg:WFoperations} consumes unbounded storage,
in the next section we prove a lower bound on the storage
required by \emph{any} protocol.

\begin{figure}
        \centering
        \begin{subfigure}[b]{0.35\textwidth}
                \includegraphics[width=\textwidth]{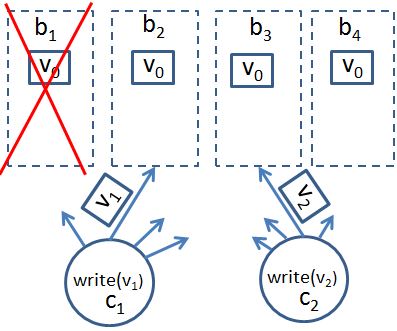}
                \caption{Clients $c_1$ and $c_2$ invoke writes.}
                \label{fig:example_a}
        \end{subfigure}%
      ~~~ 
        \begin{subfigure}[b]{0.37\textwidth}
                \includegraphics[width=\textwidth]{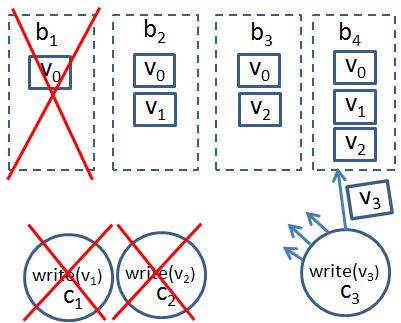}
                \caption{Clients $c_1$ and $c_2$ fail, $c_3$
                invokes write.}
                \label{fig:example_b}
        \end{subfigure}
          \caption {Example scenarios of
          erasure coded regular storage; $n=4$, $f=1$, and $k=2$.
          Small boxes represent pieces of
          the written value. Complete
          arrows represent RMWs that took effect, and short
          arrows represent pending ones.}
          \label{fig:example}
\end{figure}

\newcommand{\enc}{\mathcal{E}}
\newcommand{\dec}{\mathcal{D}}
\newcommand{\numbits}{\ell}
\section{Storage Lower Bound}
\label{sec:impossibility}

We now show a lower bound on the required storage
of any lock-free algorithm that simulates 
weakly regular MWMR register. 
Our bound stipulates that if the
number of clients that can invoke \emph{write} operations is
unbounded, then either (1) there is a time during which there
exist $f+1$ base objects each of which stores at least $D$ bits
of some \emph{write}, or (2) the storage can grow without bound.  
\newline

\textbf{Information theoretic storage model}
The storage lower bound presented in this section is obtained under a precise and natural information theoretic model of storage cost. We model the general behavior of a base object in a distributed protocol as follows. Upon each RMW operation triggered on it, the base object implements some function $\enc$, whose inputs are the values currently stored in the base object and the data provided with the write. After the RMW operation, the bits output from $\enc$ are everything that is stored in the base object. Upon a read operation triggered on the base object, the bits currently stored in it are input to some function $\dec_i$, whose output is the value returned to the reader. To justify this model, let us observe that the role of base objects in the distributed register emulation is to store sufficient information to guarantee successful information reconstruction by a client following some future read. In the next lemma we give a more formal definition of the functions $\enc$ and $\dec_i$, and prove an elementary lower bound on the number of bits that $\enc$ needs to output.
\begin{lemma}
\label{lem:pigeonhole}
Let $\enc$ be a function on $s$ arguments $u_1,\ldots,u_s$
taking values from sets $\mathbb{U}_1,\ldots,\mathbb{U}_s$,
respectively. Let the output of $\enc$ be a binary vector
$\{0,1\}^{\numbits}$. If there exist $s$ functions
$\{\dec_i\}_{i=1}^{s}$ such that
$\dec_i(\enc(u_1,\ldots,u_s))=u_i$ for every assignment to
$u_1,\ldots,u_s$, then necessarily
$\numbits\geq\left\lceil\log_2
(|\mathbb{U}_1|\cdot\ldots\cdot|\mathbb{U}_s|)\right\rceil$.
\end{lemma}
\begin{proof}
By a simple pigeonhole argument. For simplicity we assume that
the sizes $|\mathbb{U}_i|$ are powers of $2$ for every $i$.
Suppose the theorem statement is not true, that is, the output
of $\enc$ has fewer than
$\log_2(|\mathbb{U}_1|\cdot\ldots\cdot|\mathbb{U}_s|)$ bits.
Then there exist at least two assignments to $u_1,\ldots,u_s$
that map to the same output of $\enc$. Hence the outputs of the
functions $\{\dec_i\}_{i=1}^{s}$ will be the same on both
assignments, which is a violation because at least one $u_i$
differs between the two assignments.
\end{proof}
We next show how Lemma~\ref{lem:pigeonhole} implies lower bounds
on the storage used in base objects. Since the information
reconstruction algorithm is run by the client on inputs from
base objects, we may regard each RMW operation $i$ as requiring
the base object to store a value $u_i$ from some set
$\mathbb{U}_i$. The size of the set $\mathbb{U}_i$ may change
arbitrarily between writes and base objects. The particular
choices of set sizes are immaterial for the current discussion,
but in general they satisfy the necessary condition that
globally on all surviving base objects the product of set sizes
is at least $|\mathbb{V}|$. In the next lemma we prove that the
most general function implemented by a base object upon RMW is a
function $\enc$ as specified in Lemma~\ref{lem:pigeonhole}.
\begin{lemma}
\label{lem:hardcode}
Without loss of generality, a function $\enc$ used by a base
object is a fixed (``hard coded'') function that does not depend
on the instantaneous values $u_1,\ldots,u_s$.    \end{lemma}
\begin{proof}
Suppose the base object has a family of functions
$\enc^{1},\ldots,\enc^{m}$ that each maps values
$u_1,\ldots,u_s$ to bits. Then, in order to allow recovering the
$u_i$ values, we must store additional $\log_2(m)$ bits to
inform the functions $\dec_i$ about which $\enc^{j}$ function
was used. Therefore, this scenario is equivalent to having
$\enc(u_1,\ldots,u_s)=[\enc^{j}(u_1,\ldots,u_s);j]$, where $;$
represents concatenation, and $\enc$ is a fixed function.
\end{proof}
Lemma~\ref{lem:hardcode} addresses the possibility of base
objects to reduce the amount of storage by adapting their
functions to the instantaneous stored values. The lemma proves
that without prior knowledge on the written data it is not
possible to adaptively reduce the storage requirement mandated
by Lemma~\ref{lem:pigeonhole}. Now we are ready to prove the
main property needed for our storage model. The next theorem
shows that each write to a base object {\em must} add a number
of bits depending on the required set size for that write,
irrespective of the information presently stored from prior
writes.
\begin{theorem}
\label{thm:additive}
Any write triggered on a base object wih value
$u_s\in\mathbb{U}_s$ adds at least $\log_2(|\mathbb{U}_s|)$ bits.
\end{theorem}
\begin{proof}
We prove by induction on $s$. By the induction hypothesis after
$s-1$ writes the base object stores
$\log_2(|\mathbb{U}_1|\cdot\ldots\cdot|\mathbb{U}_{s-1}|)$ bits.
Then following write $s$ triggered on the base object, we know
from Lemmas~\ref{lem:pigeonhole},\ref{lem:hardcode} that any
function implemented in the base object that will allow
recovering $u_1,\ldots,u_{s}$ needs at least
$\log_2(|\mathbb{U}_1|\cdot\ldots\cdot|\mathbb{U}_{s}|)$ bits.
By simple subtraction we get that the new write adds at least
$\log_2(|\mathbb{U}_s|)$ bits.
\end{proof}
The outcome from Theorem~\ref{thm:additive} is that the base
object storage cost in bits is obtained as the sum of the
storage requirements of individual writes. Hence in the sequel
we can assume without loss of generality that {\em each stored
bit is associated with a particular write}.

With the storage model in place, we now
organize the proof  as follows: First, in
Observation \ref{obs:cannotComplete}, we observe a necessary
condition for a write operation to complete.
Next, we define an (unfair)
adversary, and in Lemma \ref{lem:secret}, we show that under this
adversary's behavior, no write operation can complete as long as
the number of base objects that store at least $D$ bits that are
associated with some written value is less than $f$. 
Finally, in Lemma 
\ref{lem:impInvoke} and Theorem \ref{lem:theorem} we show that
for every size $S$, for any algorithm that uses less storage than
$S$ and with which the number of base objects that
store at least $D$ bits of some written value is
less than $f$ at a given time we can build a fair run in
which no write operation completes.  

For any time $t$ in a run $r$ of
an algorithm $A$ we define the following sets, as illustrated in
Figure \ref{fig:definitions}.

\begin{itemize}

  \item $C(t)$: the set of all clients that have outstanding
  write operations at time $t$.
  
  \item $C^+(t) \subseteq C(t)$: the set of clients 
  that have outstanding write operations $write_i(v_i)$ s.t.\
  at least one bit associated with $v_i$ is stored in one of the
  base objects or in one of the other correct clients at time
  $t$.

  \item $C^-(t) = C(t) \setminus C^+(t)$. Clients in $C^-(t)$ may
  have attempted to store a bit via an RMW that did not respond,
  or may have stored information that was subsequently erased,
  or may have not attempted to store anything yet.
  
  \item $F(t) = \{ b_i \in N \mid b_i$ stores $D$ bits of some
  write at time $t$ $\}$.
  
\end{itemize}

\begin{figure}
        \centering
        \begin{subfigure}[b]{0.42\textwidth}
                \includegraphics[width=\textwidth]{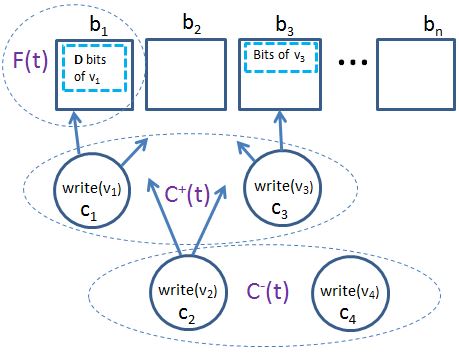}
                \caption{Time $t$}
                \label{fig:timet}
        \end{subfigure}%
      ~~~ 
        \begin{subfigure}[b]{0.42\textwidth}
                \includegraphics[width=\textwidth]{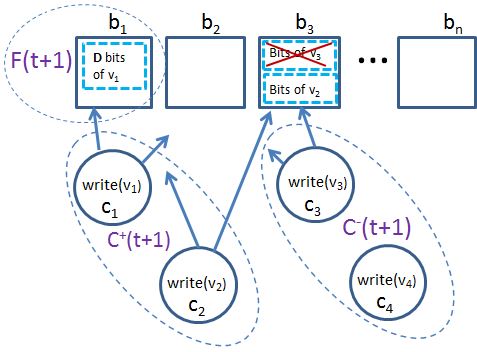}
                \caption{Time $t+1$}
                \label{fig:timet1}
        \end{subfigure}
          \caption{Example run of a storage algorithm. Clients
          $c_1,\ldots,c_4$ have outstanding writes.}
          \label{fig:definitions}
\end{figure}

\noindent 
From the definition of $C^+(t)$ we get the following:

\begin{observation}
\label{obs:strangeValues}

At any time $t$ in a run $r$, the storage size is at least
$|C^+(t)|$ bits.

\end{observation}


\begin{observation}
\label{obs:cannotComplete}

Consider a run $r$ of an algorithm that simulates a weakly
regular lock-free MWMR register, and a write operation $w$ in $r$.
Operation $w$ cannot return until there is time $t$ s.t.\ for
every $B \subset N$ s.t.\ $|B| = n -f$, there is some client in
$C(t)$ whose pending write's value can be restored from $B$.

\end{observation}

\begin{proof}

Assume that some \emph{write} completes
when there is a set $B\subset N$ s.t.\ $|B|=N-f$ and there is no
client in $C(t)$ whose write's value can be restored from $B$.
Now, let all the base objects in $N \setminus B$ and all the
clients in $C(t)$ fail, and invoke a read operation $rd$. By
lock-freedom, $rd$ completes, although no value
satisfying weak regularity can be returned. A contradiction.
\end{proof}

\noindent For our lower bound, we define a particular
environment behavior that schedules actions in a way that
prevents progress:

\begin{definition}

($Ad$) 
At any time $t$, $Ad$ schedules an action as follows:
\begin{enumerate}
  \item If there is a pending RMW on a base object in $N
  \setminus F(t)$ by a client in $C^-(t)$, then choose the
  longest pending of these RMWs, allow it to take effect on the
  corresponding base object, and schedule its response.
  \item Else, choose in round robin order a client $c_i \in \Pi$
  that wants to trigger an RMW, and schedule $c_i$'s action
  without allowing it to affect the base object yet.
\end{enumerate}
In other words, $Ad$ delays RMWs triggered by clients in
$C^+(t)$ as well as RMWs on base objects in $F(t)$, and fairly
schedules all other actions. 
Thus, though this behavior may be unfair, in every
infinite run of $Ad$, every correct client gets infinitely many
opportunities to trigger RMWs. 
We demonstrate $Ad$'s behavior
in Figure \ref{fig:definitions}. 
(a) Clients $c_2$ and $c_4$ are in $C^-(t)$ at time $t$,
where $c_4$ has no pending RMWs and $c_2$ has one triggered RMW
on $b_1 \in F(t)$ and one triggered RMW on $b_3 \not\in F(t)$.
Therefore, by the first rule, $Ad$ schedules the
response on the RMW triggered by $c_2$ on $b_3$. 
(b) In this example $c_2$ overwrites $b_3$ and so
$c_3$ moves from $C^+$ to $C^-$. Since $c_3$ is the only client
that has a pending RMW on a base object not in $F(t+1)$, $Ad$
schedules the response on the RMW triggered by $c_3$
on $b_2$ at time $t+1$. 
Now notice that at time $t+2$ there is no
client in $C^-(t+2)$ with a pending RMW on a base object in $N
\setminus F(t+2)$, and thus, by the second rule, $Ad$ chooses in
round robin a client in $\Pi$ and allows it to trigger an RMW.

\end{definition}

\noindent The following observation immediately follows
from the adversary's behavior.

\begin{observation}
\label{obs:fconst}

Assume an infinite run $r$ in which the environment behaves like
$Ad$. For each base object $bo$, if $bo \in F(t)$ at
some time $t$, then $bo \in F(t')$ for all $t' > t$.

\end{observation} 

\noindent Another consequence of $Ad$'s behavior is captured by
the following:

\begin{lemma}
\label{lem:secret}

As long as the environment behaves like $Ad$, for any
time $t$ when $|F(t)| \leq f$,
there is a set $B$ of $n-f$ base objects s.t.\ there
is no client in $C(t)$ whose value can be restored from $B$ at
time $t$.

\end{lemma}

\begin{proof}

As soon as a client $c_i$ stores a piece of data in a base
object, $c_i$ joins $C^+$, and from that point on, as long as
its data remains in the system, $c_i$ is prevented by Ad from
storing any further values. Therefore, unless $c_i$ stores all
$D$ bits of its value in some base object, it is
impossible to reconstruct this value from the bits that were
stored. Since the number of base objects storing all $D$ bits of
some client value at some time $t$ is no more than $|F(t)|$, and
since $|F(t)| \leq f$, the lemma follows.

\end{proof}

\noindent From Observation \ref{obs:cannotComplete} and Lemma
\ref{lem:secret} we conclude:

\begin{corollary}
\label{cor:cannotcomplete}

Consider a run $r$ of an algorithm that simulates a weakly
regular lock-free MWMR register.
If the adversary behaves like $Ad$, and $|F(t)| \leq f$ for all
$t$ in $r$, then no write completes in $r$.

\end{corollary}

Having shown that adversary $Ad$ can prevent progress in
algorithms that store $D$ bits of information in too few  
base objects, we turn to show that we can prevent
progress also in fair runs, leading to violation of
lock-freedom.

\begin{lemma}
\label{lem:impInvoke}

Consider a finite run $r$ with $t$ steps of an algorithm that
simulates a lock-free MWMR register, where the
environment behaves like adversary $Ad$.
If $C^-(t) \neq \{\}$ and $|F(t)| \leq f$, then it is possible to
extend $r$ by allowing the environment to
continue to behave like $Ad$ up to a time $t' \geq t $ when
either $|F(t')|>f$ or some client $c_i \in C(t')$ either returns
(i.e., completes the write) or receives a response from some base
object.
\end{lemma}

\begin{proof}

Consider a client $c_i \in C^-(t)$, and denote by $T_{c_i}(t)$
the set of base objects on which $c_i$ has pending RMWs
at time $t$.
We first show that if $c_i$ neither receives a response from any
base object nor returns, we can extend $r$ to some time $t''$
s.t.\ $|T_{c_i}(t'')|>f$ at time $t''$.

We extend $r$ by allowing the
environment to continue to behave like $Ad$ until the first time
$t'$ in which $c_i$ is the next client chosen by the
adversary to trigger an RMW.
If $c_i$ receives a response from some base object by time $t'$,
we are done.
Else, by definition of $Ad$, $T_{c_i}(t') \subseteq 
F(t') $.
Now consider a fair run $r'$ that is
identical to $r$ till time $t'$, and at time $t'$ all the
clients except $c_i$ fail.
Notice that $|T_{c_i}(t')| \leq |F(t')| \leq f$, so
$c_i$ cannot wait for responses from base objects in
$T_{c_i}(t')$, and therefore, by lock-freedom, $c_i$ either
returns, or triggers an RMW on some base object in $N\setminus
T_{c_i}(t')$ at time $t'$ in $r'$.
The runs $r$ and $r'$ are indistinguishable to $c_i$,
hence, $c_i$ either returns or triggers an RMW
on some base object in $N\setminus T_{c_i}(t')$ at time $t'$ in
$r$.
If $c_i$ returns we are done.

We repeat this extension several times until, (after at
most $f+1$ times), at some time $t''$, $|T_{c_i}(t'')|>f$. If
$|F(t'')|>f$, we are done. Otherwise, $T_{c_i}(t'')
\not\subseteq F(t'')$, and therefore, $Ad$ schedules a response
to one of the pending RMWs of $c_i$ at time $t''$.

\end{proof}

\begin{theorem}
\label{lem:theorem}

For any $S$, there is no 
algorithm that simulates a weakly regular lock-free MWMR register
with less storage than $S$ s.t.\
at every time $t$, $|F(t)| \leq f$.

\end{theorem}

\begin{proof}
\label{theorem:storage}

Assume by way of contradiction that there is such an algorithm,
$A$. We build a run of $A$ in which the
environment behaves like adversary $Ad$.

We iteratively build a run $r$ with infinitely many responses,
starting by invoking $S$ write operations and allowing the run
to proceed according to $Ad$ until some time $t$.
By the assumption, the storage is less than
$S$, so by Observation \ref{obs:strangeValues}, $|C^+(t)| <S$,
and since $|C(t)|=S, $ $C^-(t) \neq \{\}$.
Now since $|F(t)| \leq f$, by Lemma \ref{lem:impInvoke}, we can
extend $r$ to a time $t'$, where the environment behaves like
$Ad$ until time $t'$ and some client $c_i \in C^-(t')$ either
returns or receives a response from some base object at time $t'$.
By Corollary \ref{cor:cannotcomplete}, $c_i$ does not
return, and thus, it receives a response.

By repeating this process, we get a run $r$ with infinitely many
responses. By Observation \ref{obs:fconst}, and by the assumption
that $|F(t)| \leq f$, there is a time $t_1$ in $r$ s.t.\ for any
time $t_2>t_1$, $F(t_1)=F(t_2)$.
Notice that by the adversary's behavior, each correct
client gets infinitely many opportunities to trigger RMWs. In
addition, since $Ad$ picks responses from base objects not in
$F(t)$ in the order they are triggered, every client that
receives infinitely many responses, receives a response to every
RMW it triggers on a base object in $N \setminus F(t_1)$.
Therefore, we can build a fair run $r'$ that is identical to $r$
but every base object $bo \in F(t_1)$ fails at time $t_1$, and
every client that receives finitely many responses fails after
its last response. Since there are infinitely many responses in
$r'$ and the number of clients invoking operations in this
run is finite, there is at least one client that receives
infinitely many responses in $r'$, and thus is correct in $r$.
Therefore, by lock-freedom, some client eventually completes its
write operation in $r'$. Since $r$ and $r'$ are
indistinguishable to all clients and base objects that are
correct in both, the same is true in $r$. A contradiction to
Corollary \ref{cor:cannotcomplete}.

\end{proof}

\noindent From Theorem \ref{lem:theorem}, it follows that if the
storage is bounded, then there is a time in which
$f+1$ base objects store $D$ bits of some \emph{write}. This
yields the following bound:

\begin{corollary}

There is no algorithm that
simulates a weakly regular lock-free MWMR register and stores
less than $(f+1)D$ bits in the worst case. 

\end{corollary}

\vspace*{-.3cm}
\section{Strongly Regular MWMR Register Emulation}
\label{sec:WRAlgorithms}

We present a storage algorithm that combines full replication
with erasure coding in order to achieve the advantages
of both.
The main idea behind our algorithm is to have base objects store
pieces from at most $k$ different \emph{writes}, and then turn
to store full replicas.
In Appendix \ref{AppSub:WR}, we prove
the following about our algorithm:\\

\begin{thmclone}{theorem:WR}
There is an FW-terminating algorithm that
simulates a strongly regular register, whose storage is bounded
by $(2f+k)2D$ bits, and in runs with at most $c < k$ concurrent
writes, the storage is bounded by $(c+1)D/k$ bits. Moreover, in
a run with a finite number of writes, if all the writers are
correct, the storage is eventually reduced to $(2f+k)D/k$ bits.\\
\end{thmclone}

\textbf{Data structure}~~~  
The algorithm uses the same definitions as the safe one
(Section \ref{sec:WFsafe}), given in
Algorithm~\ref{alg:WFdefinitions}, and
its pseudocode appears in Algorithms \ref{alg:Operations} and
\ref{alg:weakProcedures}. The algorithm relies on a set of $n$
shared base objects $bo_1,\dots,bo_n$ each of which consists of
three fields $V_p$, $V_f$, and $storedTS$:

\vspace*{-.3cm}
\begin{framed}
\begin{center}
$bo_i = \langle storedTS, V_p,V_f \rangle$ s.t.\ $V_f,V_p
\subset Chunks$, and $storedTS \in TimeStamps$,\\
initially $\langle \langle 0,0
\rangle, \{\langle \langle 0,0 \rangle , \langle v_{0_i}, i
\rangle \rangle\} , \{\} \rangle$.
\end{center}
\end{framed}

\vspace*{-.3cm}
The $V_p$ field holds a set of timestamped coded pieces of
values so that the $i^{th}$ piece of any value can only be
stored in the $V_p$ field of object $bo_i$. The $V_f$ field
stores a timestamped replica of a {\em single}\/ value, (which
for simplicity is represented as a set of $k$ coded pieces). And
$storedTS$ holds the highest timestamp of a write that is known
to this object to have completed the update round on $n-f$ base
objects (see below).

\textbf{Write operation and storage efficiency}~~~ 
The write operation (lines
\ref{line:reWbegin}--\ref{line:reWend}) consists of 3
sequentially executed rounds: {\em read timestemp}, {\em
update}, and {\em garbage collection}; and, the read consists of
one or more sequentially executed {\em read}\/ rounds.  
At each round, the client invokes RMWs on all base objects in
parallel, and awaits responses from at least $n-f$ base objects.
The read rounds of both write and read rely on the readValue
routine (lines
\ref{line:regularRVbegin}--\ref{line:regularRVend}) to collect
the contents of the $V_p$ and $V_f$, fields stored at $n-f$ base
objects as well as to determine the highest $storedTS$ timestamp
known to these objects.
The implementations of the update and garbage collection
rounds are given by the update (lines
\ref{line:updateBegin}--\ref{line:updateEnd}) and GC (lines
\ref{line:freeBegin}--\ref{line:freeEnd}) routines, respectively.

The write
implementation starts by breaking the supplied value $v$ into
$k$ erasure-coded pieces (line \ref{line:reWencode}).
This is followed by invoking the read round where the client
uses the combined contents of the $V_p$, $V_f$ and $storedTS$
fields returned by readValue to determine the timestamp $ts$ to
be stored alongside $v$ on the base object. This timestamp is
set to be higher than any other timestamp that has been returned
(line \ref{line:reWts}) thus
ensuring that the order of the timestamps associated with the
stored values is compatible with the order of their
corresponding writes (which is essential for regularity).

The client then proceeds to the update round where it attempts
to store the $i^{th}$ coded piece $\langle e, i \rangle$ of $v$
in $bo_i.V_p$ if the size of $bo_i.V_p$ is less than $k$ (lines
\ref{line:storeOnePiece}), or its full replica in $bo_i.V_f$ if
$ts$ is higher than the timestamp associated with the value
currently stored in $bo_i.V_f$ (line \ref{line:StorekPieces}).
Note that storing $\langle e, i \rangle$ in $bo_i.V_p$ coincides
with an attempt to reduce its size by removing stale coded
pieces of values whose timestamps are smaller than $storedTS$
(line \ref{line:storeOnePiece}).
This guarantees that the size of $V_p$ never exceeds the number
$c < k$ of concurrent writes, which is a key for achieving our
adaptive storage bound. Lastly, the client updates
$bo_i.storedTS$ so as its new value is at least as high as the
one returned by the readValue routine.
This allows the timestamp associated with the latest complete
update to propagate to the base object being written, in order
to prevent future writes of old pieces into this base object.

In the write's garbage collection round, the client attempts to
further reduce the storage usage by (1) removing all coded
pieces associated with timestamps lower than $ts$ from both
$bo_i.V_p$ and $bo_i.V_f$ (lines
\ref{line:keepnew1}--\ref{line:keepnew2}), and (2) replacing a full replica (if
it exists) of its written value $v$ in $bo_i.V_f$ with its
$i^{th}$ coded piece $\langle e, i \rangle$ (line
\ref{line:keepone}).
It is safe to remove the full replica and values with older
timestamps at this point, since once the update round has
completed, it is ensured that the written value or a newer
written value is restoreable from any $n-f$ base objects.
This mechanism ensures that all coded pieces except the ones
comprising the value written with the highest timestamp are
eventually removed from all objects' $V_p$ and $V_f$ sets, which
reduces the storage to a minimum in runs with
finitely many writes, which all complete. 
The
garbage collection round also updates the $bo_i.storedTS$ field
to ensure its value is at least as high as $ts$, reflecting the
fact that a write with $ts'>ts$ that the update
round.

\textbf{Key Invariant and read operation}~~~ 
The write implementation described above guarantees the
following key invariant: 
at all times, a value written by either the latest complete write
or a newer write is available from every set consisting of at
least $n-f$ base objects
(either in the form of $k$ coded pieces in the objects' $V_p$
fields, or in full from one of their $V_f$ fields). 
Therefore, a read will always be able to reconstruct the latest
completely written or a newer value provided it can successfully
retrieve $k$ matching pieces of this value.
However, 
a read round may sample different base objects at different
times (that is, it does not necessarily obtain a snapshot of all
base objects), and 
the number of pieces stored in $V_p$ is
bounded. Thus, the read may be unable to see $k$
matching pieces of any single new value for indefinitely long,
as long as new values continue to be written concurrently
with the read. 

To cope with such situations, the reads are only required
to return in runs where a finite number of writes are invoked,
thus only guaranteeing FW-Termination. Our implementation of
read (lines \ref{line:reRbegin}--\ref{line:reRend}) proceeds by invoking multiple consecutive
rounds of RMWs on the base objects via the readValue routine.
After each round, the reader examines the collection of the
values and timestamps returned by the base objects to determine
if any of the values having $k$ matching coded pieces are
associated with timestamps that are at least as high as
$storedTS$ (line \ref{line:SRwhile}).
If any such value is found, the one associated with the highest timestamp is
returned (line \ref{line:reReturn}).
Otherwise, the reader proceeds to invoke another round of base
object accesses.
Note that returning values associated
with older timestamps may violate regularity, since they may
have been written earlier than the write with timestamp
$storedTS$, which in turn may have completed before the read
was invoked.

\begin{algorithm}[H]

 \caption{Strongly regular register emulation. Algorithm for
 client $c_j$.}
 \label{alg:Operations} 
\begin{algorithmic}[1]
\small

\State \textbf{local variables:} 
\State  \hspace*{0.5cm} $storedTS,ts \in TimeStamp$, $WriteSet
\in Pieces$


\Operation{$Write(v)$}{}
\label{line:reWbegin}

\State $WriteSet\leftarrow \emph{encode}(v)$   
\label{line:reWencode}

\State $ \langle \emph{storedTS} ,ReadSet \rangle \leftarrow
\emph{readValue}()$
\Comment{\textbf{round 1:} read timestamps}

\State $n \leftarrow max(storedTS.num,~ max\{n' \mid
\langle\langle n',*\rangle, * \rangle \in ReadSet \})$
\label{line:reWts}
\State $ts \leftarrow \langle n+1,j \rangle$

\State \textbf{$||$ for} i=1 to n 
\Comment{\textbf{round 2:} update}
\State \hspace*{0.5cm}
$update(bo_i,\emph{WriteSet},ts,\emph{storedTS},i)$

\State \textbf{wait} for  $n-f$ responses

\State \textbf{$||$ for} i=1 to n 
\Comment{\textbf{round 3:} garbage collect}

\State \hspace*{0.5cm}  $\emph{GC}(bo_i,\emph{WriteSet},
ts,i)$  
\State  \textbf{wait} for $n-f$ responses

\State \textbf{return} ``ok''

\EndOperation
\label{line:reWend}

\Operation{$Read()$}{}
\label{line:reRbegin}

\State $\langle \emph{storedTS}, ReadSet \rangle \leftarrow
\emph{readValue}()$
\State \textbf{while}  $\nexists ts \geq \emph{storedTS}$ s.t.\
$|\{\langle ts,v \rangle \mid \langle ts,v \rangle \in ReadSet
\} | \geq k$ 
\label{line:SRwhile}
	\State \hspace*{0.5cm} $\langle \emph{storedTS}, ReadSet
	\rangle \leftarrow \emph{readValue}()$

\State $ts' \leftarrow \max\limits_{\raise3pt\hbox{$\scriptstyle
ts \geq \emph{storedTS}$}}{}( |\{\langle ts,v \rangle \mid \langle ts,v
\rangle \in ReadSet \} | \geq k)$
\label{line:SR17}

\State \textbf{return} $\emph{decode} ( \{v \mid
\langle ts',v \rangle \in ReadSet \})$
\label{line:reReturn}

\EndOperation
\label{line:reRend}

\algstore{bkbreak}

\end{algorithmic}
\end{algorithm}

\vspace*{-.5cm}

\begin{algorithm}[H]
 \caption{Functions used in strongly regular register
 emulation.}
 \label{alg:weakProcedures} 
\begin{algorithmic}[1]
\algrestore{bkbreak}
\small

\Procedure{$readValue()$}{}
\label{line:regularRVbegin}

\State $\emph{ReadSet} \leftarrow \{\}$, $T \leftarrow \{\}$

\State \textbf{$||$ for} i=1 to n
\State \hspace*{0.5cm} $tmp \leftarrow read(bo_i)$
\State \hspace*{0.5cm} $ReadSet \leftarrow ReadSet \cup
tmp.V_f \cup tmp.V_p$
\State \hspace*{0.5cm} $T \leftarrow T \cup \{ tmp.storedTS\}$

\State \textbf{wait} for  $n-f$ responses
\State \textbf{return}    $ \langle max(T),\emph{ReadSet}
\rangle$

\EndProcedure
\label{line:regularRVend}

\State \textbf{update}$(bo,\emph{WriteSet},ts,\emph{storedTS},i)
\triangleq$
\label{line:updateBegin}

\State \hspace*{0.5cm} \textbf{if} $ts \leq bo.\emph{storedTS}$
\State \hspace*{1cm} \textbf{return}
\State \hspace*{0.5cm} \textbf{if} $|bo.V_p|< k $ 
\Comment{write a piece and remove old pieces}

\State \hspace*{1cm} $bo.V_p  \leftarrow  bo.V_p \setminus
\{ \langle ts',v \rangle \in bo.V_p \mid ts' < \emph{storedTS}
\} \cup \{\langle ts, \langle e,i \rangle \rangle \mid \langle
e, i \rangle \in \emph{WriteSet} \}$
\label{line:storeOnePiece}

\State \hspace*{0.5cm} \textbf{else if}  $bo.V_f = \{\} ~ \vee$
$ \exists ts' < ts : \langle ts', * \rangle \in bo.V_f$
\Comment{write a full replica}

\State \hspace*{1cm} $bo.V_f \leftarrow \{\langle ts, \langle
e,j \rangle  \rangle \mid  \langle  e,j \rangle \in
\emph{WriteSet} ~ \wedge j \in \{1,\ldots, k\} \}$
\label{line:StorekPieces}

\State \hspace*{0.5cm} $bo.\emph{storedTS} \leftarrow
max(bo.\emph{storedTS},\emph{storedTS})$
\label{line:updateEnd}

\State \textbf{GC}$(bo,\emph{WriteSet},ts,i)$ $\triangleq$ 
\label{line:freeBegin} 

\State \hspace*{0.5cm} $bo.V_p \leftarrow \{ \langle 
ts',v\rangle \in bo.V_p | ts' \geq ts \}$
\label{line:keepnew1} 
\Comment{keep only new pieces}
\State \hspace*{0.5cm} $bo.V_f \leftarrow \{ \langle 
ts',v\rangle \in bo.V_f | ts' \geq ts \}$
\label{line:keepnew2} 

\State \hspace*{0.5cm} \textbf{if} $\langle ts,*\rangle \in
 bo.V_f$
 \Comment{if $V_f$ holds a full replica of my write}
  
\State \hspace*{1cm} $bo.V_f \leftarrow \{\langle
ts, \langle e,i \rangle \rangle \mid \langle e, i  \rangle
\in \emph{WriteSet} \}$
\label{line:keepone} 
\Comment{keep only one piece of it} 

\State \hspace*{0.5cm} $bo.\emph{storedTS} \leftarrow
max(bo.\emph{storedTS},ts)$
\label{line:freeEnd} 

\Statex

\end{algorithmic}
\end{algorithm}

\section{Discussion}
\label{sec:discussion}

We studied the storage cost of 
shared register simulations in asynchronous fault-prone shared
memory. We proved a lower bound on the required storage  
of any lock-free algorithm that simulates a
weakly regular MWMR register. Our bound
stipulates that if write concurrency is unbounded, then
either (1) there is a time during which there exist $f+1$ base
objects each of which stores a full replica of some
written value, or (2) the storage can grow without
bound.

We showed that our lower 
bound does not hold for safe register emulation. 
And finally, by understanding these inherent limitations, 
we introduced a new technique for emulating shared storage by
combining full replication with erasure codes.
We presented an implementation of
an FW-Terminating strongly regular MWMR register, whose storage
cost is adaptive to the concurrency level of write operations up
to certain point, and then turns to store full replicas. In
periods during which there are no outstanding writes, our
algorithm's storage cost is reduced to a minimum.

Our work leaves some questions open for future work. First, we
conjecture that a wait-free implementation with similar storage
costs requires readers to write. Second, our algorithm requires
more storage than the bound. 
We believe that our technique can be used for implementing
additional adaptive algorithms, with storage costs
closer to the lower bound. 
Another 
interesting question that remains open is whether the
liveness condition of the lower bound is tight. In other words,
is there an algorithm that emulates an obstruction-free weakly
regular register with a better storage cost.

\newpage
\appendix
\section{Correctness Proofs}
\label{App:AppendixA}

\subsection{Wait-Free and Safe Algorithm}
\label{AppSub:WF}

Here we prove the algorithm in Section \ref{sec:WFsafe}.

\begin{lemma}
\label{lem:WFstorage}

The storage of the algorithm is $nD/k$.

\end{lemma}

\begin{proof}

The size of each piece is $D/k$. We have $n$ base objects, and
each base object stores exactly one piece.

\end{proof}

\begin{lemma}
\label{lem:WFwaitfree}

The algorithm is wait-free.

\end{lemma}

\begin{proof}

There are no loops in the algorithm, and the only blocking
instructions are the waits in lines
\ref{line:WFwait1} and \ref{line:WFwait2}. In both cases, clients
wait for no more than $n-f$ responses, and since no more than
$f$ base objects can fail, clients eventually continue.
Therefore, a client that gets the opportunity to perform
infinitely many actions completes its operations.

\end{proof}

We now prove that the algorithm satisfies strongly safety. We
relay on the following single observation.

\begin{observation}
\label{obs:WFmon}

The timestamps in the base objects are monotonically increasing.  

\end{observation}

\begin{definition}

For every run $r$, we define the sequential run $\sigma_{w_r}$ as
follows: All the completed write operations in $r$ are ordered in
$\sigma_{w_r}$ by their timestamp.

\end{definition}

\begin{lemma}
\label{lem:WFwrites}

For every run $r$, the sequential run $\sigma_{w_r}$ is a
linearization of $r$.

\end{lemma}

\begin{proof}

Since $\sigma_{w_r}$ has no read operations, the sequential
specification is preserved in $\sigma_{w_r}$. 
Thus, we left to show the
real time order: 
For every two completed \emph{writes}
$w_i$, $w_j$ in $r$, we need to show that if $w_i \prec_r w_j$,
then $w_i \prec_{\sigma_r} w_j$.
 
Denote $w_i$'s timestamp by $ts$.
By Observation \ref{obs:WFmon}, at any point after
$w_i$'s return, at least $n-f$ base objects store timestamps
bigger than or equal to $ts$. When $w_j$ picks a timestamp, it
chooses a timestamp bigger than those it reads from $n-f$ base
objects. Since, $n>2f$, $w_j$ picks a timestamp bigger than $ts$,
and therefore $w_j$ is ordered after $w_i$ in $\sigma_{rd}$.

\end{proof}

\begin{definition}

For every run $r$, for every \emph{read} $rd$ that has no
concurrent \emph{write} operations in $r$, we define the
sequential run $\sigma_{r_{rd}}$ by adding $rd$ to $\sigma_{w_r}$
after all the writes that precede it in $r$.

\end{definition}

In order to show that the algorithm simulates a safe
register, we proof in Lemmas \ref{lem:WFtime} and
\ref{lem:WFspec} that the real time order and sequential
specification respectively, are preserved in $\sigma_{r_{rd}}$.

\begin{lemma}
\label{lem:WFtime}

For every run $r$, for every read $rd$ that has no concurrent
write operations in $r$, $\sigma_{r_{rd}}$ preserves
$r$'s operation precedence relation (real time order).

\end{lemma}

\begin{proof}

By Lemma \ref{lem:WFwrites}, the order between the writes in
$\sigma_{r_{rd}}$ are preserved, and by construction of
$\sigma_{rd}$ the order between $rd$ and write operations is also
preserved.

\end{proof}

\begin{lemma}
\label{lem:WFspec}

Consider a run $r$ and any read $rd$
that has no concurrent writes in $r$. 
Then $rd$ returns the
value written by the write with the biggest timestamp
that precedes $rd$ in $r$, or $v_0$ if there is no such
\emph{write}.

\end{lemma}

\begin{proof}

In case there is no \emph{write} before $rd$ in $r$, since there
are also no writes concurrent with $rd$, $rd$ reads pieces with
timestamp $\langle 0,0 \rangle$ from all base objects, and thus,
returns $v_0$.
Otherwise, let
$w$ be the \emph{write}$(v)$ associated with the biggest
timestamp $ts$ among all the \emph{writes} invoked before $rd$
in $r$.
Let $t$ be the time when $rd$ is invoked. Recall that $rd$ has
no concurrent \emph{writes}, so all the writes invoked before
time $t$ complete before time $t$ and store there pieces in
$n-f$ base objects unless the base objects already hold a higher
timestamp.
By Observation \ref{obs:WFmon} and the fact that $w$ has the
highest timestamp by time $t$, we get that at time $t$ there are
at least $n-f$ base objects that store a piece of $v$. 
Since $n=2f+k$, every two sets of $n-f$ base objects have at
least $k$ base objects in common.
Therefore, $rd$ reads at least $k$ pieces of
$v$, and thus, restores and returns $v$.

\end{proof}

\begin{corollary}

There exists an algorithm that simulates a
safe wait-free MWMR register with a worst-case storage cost of
$nD/k=(2f/k +1)D$.

\end{corollary}

\subsection{Strongly Regular Algorithm}
\label{AppSub:WR}

Here we prove the algorithm in Section \ref{sec:WRAlgorithms}.
We start by proving the storage cost.

\begin{observation}
\label{obs:WRmon}

For every run of the algorithm, for every base object $bo_i$,
$bo_i.ts$ monotonically increasing.

\end{observation}

\begin{lemma}
\label{lem:WR1}

Consider a run $r$ of the algorithm, and two \emph{writes}
$w_1,w_2$, where $w_1$ writes with timestamp $ts_1$.
If $w_1 \prec_r w_2$,
then $w_2$ sets its $\hat{ts}$, to a timestamp
that is not smaller than $ts_1$.

\end{lemma}

\begin{proof}

By Observation \ref{obs:WRmon}, for
each base object $bo$, $bo.ts$ is monotonically increasing.
Therefore, after $w_1$ finishes the garbage collection phase,
there is a set $S$ consisting of $n-f$ base objects s.t.\ for each $bo_i \in S$, 
$bo_i.ts \geq ts$.
Recall that $n=2f + k$, thus every two sets of $n-f$ base objects
have at least one base object in common. Therefore, $w_2$ gets
a response from at least one base object in $S$ in its
first phase, and thus sets
$\hat{ts}=ts'$ s.t.\ $ts' \geq ts$.

\end{proof}

\begin{lemma}
\label{lem:SRonePiece}

For any run $r$ of the algorithm, for any base object $bo$
at any time $t$ in $r$, $bo.V_p$ does not store more than one
piece of the same write.

\end{lemma}

\begin{proof}

The \emph{writes} perform the second phase at most one
time on each base object $bo$, and in each update they store at
least one piece in $bo.V_p$. And since they does not store
in $bo.V_p$ during the third phase, the lemma follows. 

\end{proof}

\begin{lemma}
\label{lem:WRstorage1}

Consider a run $r$ of the algorithm in which the maximum number
of concurrent writes is $c<k-1$. Then the storage at any time in
$r$ is not bigger than $(2f+k)(c+1)D/k$ bits.

\end{lemma}

\begin{proof}

Recall that we assume that $n=2f+k$ and the size of each piece
is $D/k$. Thus it suffices to show that there is no time $t$ in
$r$ s.t.\ some base object stores more than
$c+1$ pieces at time $t$. 

Assume by way of contradiction that the claim is false. 
Consider the time $t$ when some $bo \in N$ stores $c+2$
pieces for the first time.
Notice that $|bo.V_p| \leq c+1 < k$ till time $t$, and
therefore, $bo.V_p$ does not contain more then one piece from
the same write, and $bo.V_f= \bot$ till time $t'$.
Now consider the write $w$ that was invoked last among all
the writes that store pieces in $bo.V_p$ at time $t$, denote
its piece by $p$.
Since $bo$ stores $c+2$ pieces at time $t'$, by Lemma
\ref{lem:WRstorage1}, there must be two writes $w_1$ and $w_2$
whose pieces $p_1,~p_2$ are stored at time $t$ in $bo.V_p$, and
both returns before $w$ is invoked.
Denote their timestamps $ts_1$ and $ts_2$, and
assume without loss of generality that $ts_1>ts_2$. By Lemma
\ref{lem:WR1}, $w$ sets its $\hat{ts}$ to $ts'$ s.t.\ $ts' \geq
ts_1 > ts_1$. Now consider two cases. First, if $p$ was
added before $p_2$, then $bo.ts>ts_2$ when $p_2$ was added.
A contradiction. Otherwise, $p$ was added after $p_2$.
Thus, $p_2$ was deleted in line \ref{line:storeOnePiece} of the
update when $p$ was added.
A contradiction.

\end{proof}

\begin{lemma}
\label{lem:WRstorage2}

The storage is never more than $(2f+k)2D$ bits at any time $t$
in any run $r$ of the algorithm.

\end{lemma}

\begin{proof}

Each base object stores no more than $2k$ pieces at any time
$t$ in $r$.
The lemma follows.

\end{proof}

\begin{lemma}
\label{lem:SWreduses}

Consider a run $r$ of the algorithm with finite number of
writes, in which all writes correct. Then the storage is
eventually reduced to $(2f+k)D/k$ bits.

\end{lemma}

\begin{proof}

Consider a \emph{write} $w$ with the biggest timestamp $ts$ in
$r$. Since $w$ is correct, and since \emph{writes} are
wait-free, $w$ returns, and eventually performs \emph{free}
on every base object. Consider a base object $bo$ s.t.\ $w$
performs \emph{free} on $bo$ at time $t$. Notice that $w$
deletes all pieces with smaller timestamps than $ts$ and set
$bo.ts=ts$ at time $t$. Now recall that $bo$ ignore all updates
with timestamp less than $bo.ts$, and therefore, $bo$ store only
$w$'s piece at any time after time $t$. The lemma follows.

\end{proof}

\noindent From Lemmas \ref{lem:WRstorage1},
\ref{lem:WRstorage2}, and \ref{lem:SWreduses} we get:
 
\begin{corollary}
\label{cor:WRstorage}

The storage of the algorithm is bounded by $(2f+k)2D$ bits,
and in runs with at most $c < k$ concurrent writes the storage
is bounded by $(c+1)D/k$ bits. Moreover, in a run with a finite
number of writes, if all the writes are correct, the storage is
eventually reduced to $(2f+k)D/k$ bits.

\end{corollary}

\noindent We no prove the liveness property.

\begin{lemma}
\label{lem:WRwaitWrite}

Consider a fair run $r$ of the algorithm. Then every
\emph{write} $w$ invoked by a correct client $c_i$ eventually
completes.

\end{lemma}

\begin{proof}

Consider a correct client $c_i$.
The \emph{write} $w$ is divided into three phase s.t.\
in each phase, $c_i$ invokes operations on all the base objects,
and waits for $n-f$ responses. The run $r$ is fair, so every
action invoked by $c_i$ on a correct base object eventually
returns, and no more than $f$ base objects fail in $r$.
Therefore, eventually $c_i$ receives $n-f$ responses in each of
the phases and returns.
 
\end{proof}

\begin{observation}
\label{obs:WRdelete}

When a piece from $bo.V_p$ is deleted, $bo.ts$ is increased.

\end{observation}

\begin{lemma}
\label{lem:SRinvHelp}

If at time $t$, $c_i$ completes the second phase of write with
timestamp $ts$, then for every $t' > t$ for every $S \subseteq N$
s.t.\ $|S| \geq n-f$, exist write $w$ with $ts' \geq ts$ s.t.\
at least $k$ pieces of $w$ are stored in $S$.
\end{lemma}

\begin{proof}

Consider time $t'$. Let $\hat{ts}$ be the highest timestamp
written by a write $w$ that completed the second phase by time
t. It is sufficient to show the lemma hold for $\hat{ts}$.

First note that $\forall bo$, $bo.ts \leq \hat{ts}$ before time
$t$, because no write with a larger timestamp than $\hat{ts}$
started the third phase.
This means that $w$'s \emph{update} left at lest one piece in
which $bo$ it occurred. Now consider a set $S$ of $n-f$ base
objects, and since $n=2f+k$, $w$'s \emph{update} occurred
in set $S'$ that contains at least $k$ base objects in $S$.

If $w$ wrote to $V_p$, it was not overwritten by time $t$,
because (1) no other write began \emph{free} with timestamp
bigger than $\hat{ts}$, and (2) since there is no base object
$bo$ s.t.\ $bo.ts \geq \hat{ts}$, no write delete $w$'s piece in
the second phase. Therefore if $w$ wrote to $V_p$ in all base
objects in $S'$, the lemma holds.

Otherwise, $w$ wrote $k$ pieces to $V_f$ in base objects in
some set $S'' \subseteq S'$.
Consider two cases: First, there is base object $bo' \in S''$
s.t.\ some write overwritten $w$'s pieces in $bo'.V_f$ before
time $t$.
Since there is no write with timestamp bigger than $\hat{ts}$
that started the third phase before time $t$, it is guarantee
that $k$ pieces with timestamp $ts'>\hat{ts}$ stored in
$bo'.V_f$ at time $t$, and the lemma holds. Else, since
$w$'s pieces stored in $S' \setminus S''$ does not overwritten
before time $t$, the lemma holds (no matter if $w$ performed the
third phase or not).

\end{proof}

\begin{invariant}
\label{inv:WR}

For any run $r$ of the algorithm, for any time $t$ in $r$, for
any set $S$ of $n-f$ base objects. 
Let $\hat{ts_s}=max \{ bo.ts \mid bo \in S\}$.
Then there is a timestamp $ts' \geq \hat{ts_s}$ s.t.\ there are
at least $k$ different pieces associated with $ts'$ in $S$.
 
\end{invariant}

\begin{proof}

We prove by induction. \textbf{Base:} the invariant holds at
time $0$. \textbf{Induction:} Assume that the induction holds
before the $t^{th}$ action is scheduled, we show that it holds
also at time $t$. Assume that the $t^{th}$ action is RMW on a
base object $bo$, and consider any set $S$ of $n-f$ base
objects.
If $bo \notin S$ then the invariant holds. Else, 
consider the two possible RMW actions: 

\begin{itemize}
  \item The $t^{th}$ action is \emph{update}.
  If no pieces are deleted, the invariant holds. If $bo.ts$ is
  increased, then consider the write with timestamp $ts$ that is
  the the biggest timestamp among all writes that
  complete the second phase before time $t$. Notice that
  $bo.ts  \leq ts$ at time $t$, and by Lemma \ref{lem:SRinvHelp},
  the invariant holds. The third option is that a piece $p$ with
  timestamp $ ts' > bo.ts$ of a \emph{write} $w$ is deleted
  and $bo.ts$ is not increased. Note that by Observation
  \ref{obs:WRdelete}, such piece can be deleted only from
  $bo.V_f$, and since $p$ is overwritten by
  $k$ pieces with bigger timestamp, the invariant
  holds.

  \item The $t^{th}$ action is \emph{free}. If $bo.ts$ is not
  changes, then the invariant holds. Else, Consider the write
  with the biggest timestamp $ts$ among
  all writes that complete the second phase before time $t$.
  Note that $bo.ts$ is set to a timestamp $ts' \leq ts$, so by
  Lemma \ref{lem:SRinvHelp}, the invariant holds.

\end{itemize}

\end{proof}

\begin{lemma}
\label{lem:WRreadLockFree}

Consider a fair run $r$ of the algorithm. If there is a finite
number of write invocations in $r$, then every
read operation $rd$ invoked by a client $c_i$ eventually
returns.

\end{lemma}

\begin{proof}

Assume by way of contradiction that $rd$ does not return in $r$.
By Lemma \ref{lem:WRwaitWrite}, the \emph{writes} are wait-free,
and since the number of \emph{write} invocations in $r$ is
finite, there is a time $t$ in $r$ s.t.\ no \emph{write}
performs actions after time $t$. Therefore, any \emph{read} that
invokes $readValue()$ procedure after time $t$ receives a set
$S$ of values that is stored in a set of $n-f$ base objects at
time $t$. By invariant \ref{inv:WR}, there is a timestamp $ts$
s.t.\ there is at least $k$ different pieces in $S$ associated
with $ts$, and $ts > bo.ts$ for all $bo \in S$.
Now since the every correct \emph{read} $rd$ invokes
$readValue()$ infinitely many times in $r$, $rd$ returns. A
contradiction.

\end{proof}

\noindent The next corollary follows from Lemmas
\ref{lem:WRwaitWrite}, \ref{lem:WRreadLockFree}. 

\begin{corollary}
\label{cor:WRliveness}

The algorithm satisfies the WF-termination property.

\end{corollary}

\noindent We now prove that the algorithm satisfies strong
regularity.

\begin{definition}

For every run $r$, $\sigma_r$ is a sequential run s.t.\ the
\emph{writes} in $r$ are ordered in $\sigma_r$ by their
timestamp, and every \emph{read} in $r$ that returns a value
associate with timestamp $ts$, is ordered in
$\sigma_r$ immediately after the \emph{write} that is associate
with timestamp $ts$.

\end{definition}

\noindent For simplicity we say the that $v_0$ was written by
\emph{write} $w_0$ that associated to timestamp $0$ at time $0$.

\begin{lemma}
\label{lem:WRregulre}

Consider a run $r$, and a read $rd$ that returns a value
$v$.
Consider also the timestamp $ts'$ that $rd$ obtains in 
line \ref{line:SR17} (Algorithm \ref{alg:Operations}). Then $v$
is the value written by a \emph{write} associated with timestamp
$ts'$ or $v_0$ if $ts'=0$.

\end{lemma}

\begin{proof}

By the code, if $ts'= 0$, then $rd$ returns $v_0$. Now notice
that $rd$ obtains at least $k$ different pieces associated with timestamp
$ts'$, thus by decode definition, $rd$ returns $v$.

\end{proof}

\begin{corollary}
\label{cor:WRspecification}

For every run $r$, $\sigma_r$ satisfies the sequential
specification. 

\end{corollary}

\begin{observation}
\label{obs:WRmax}

Consider a \emph{write} $w$ that obtains $ts$ and $\hat{ts}$ in
the first phase, then $ts> \hat{ts}$.

\end{observation}

\begin{lemma}
\label{lem:WGsecondPhase}

For every run $r$, for every two writes $w_1,w_2$ with
timestamp $ts_1,ts_2$.
If $w_2$ was invoked after $w_1$ finished the second phase, then
$ts_1 < ts_2$.

\end{lemma}

\begin{proof}

First notice that for every base object $bo$, if a \emph{write}
$w$ overwrites pieces of a \emph{write} $w'$ in $bo,V_f$, that
$w$' timestamp is bigger than $w'$'s. And by Observation
\ref{obs:WRmax}, if $w$ deletes $w'$'s piece from $bo.V_p$, then
it stores a piece with bigger timestamp than $w'$'s timestamp.
Therefore, the maximal timestamp in each base object is
monotonically increasing. 
Now recall that in
the second phase $w_1$ performed \emph{update} on $n-f$ base
object, and notice that after $w_1$ performs \emph{update} on
base object $bo$ the maximal timestamp in $bo$ is at lest as
big as $ts_1$. Now since two sets of $n-f$ base object have at
least one base object in common, $w_2$ picks $ts > ts_1$.

\end{proof}

\begin{lemma}
\label{lem:WRorderbefore1}

For every run $r$, for every two writes $w_1,w_2$ in $r$,
if $w_1 \prec_r w_2$, then $w_2$ is not ordered before $w_1$ in
$\sigma_r$.

\end{lemma}

\begin{proof}

Follows immediately from Lemma \ref{lem:WGsecondPhase}. 

\end{proof}

\begin{lemma}
\label{lem:WRorderbefore2}

For every run $r$, for every read $rd$ and write
$w_1$, if $rd \prec_r w_1$, then $w_1$ is not ordered before
$rd$ in $\sigma_r$.

\end{lemma}

\begin{proof}

Assume that $rd$ returns value that is associated with timestamp
$ts$ belonging to some \emph{write} $w$, and $w_1$ is associated
with timestamp $ts_1$.
Since $rd$ returns $w$'s value, $w$ begins the third phase
before $rd$ returns. And since $w_1$ was invoked after $rd$
returns, $w_1$ was invoked after $w$'s second phase. Therefore,
by Lemma \ref{lem:WGsecondPhase}, $ts_1 > ts$, and thus $w_1$ is
ordered after $w$ in $\sigma_r$. Recall that by the construction
of $\sigma_r$, $rd$ is ordered immediately after $w$ in
$\sigma_r$, hence, $rd$ is ordered before $w_1$ in $\sigma_r$.

\end{proof}

\begin{lemma}
\label{lem:WRorderbefore3}

For every run $r$, for every read $rd$ and write
$w_1$, if $w_1 \prec_r rd$, then $rd$ is not ordered before
$w_1$ in $\sigma_r$.

\end{lemma}

\begin{proof}

Consider a \emph{write} $w_1$ with timestamp $ts_1$ and a
\emph{read} $rd$ s.t.\ $w_1 \prec_r rd$. Assume by way of
contradiction that $rd$ is ordered before $w_1$ in $\sigma_r$.
Then $rd$ returns a value with a timestamp $ts$ that is
associated with a \emph{write} $w$ that is ordered before $w_1$
in $\sigma_r$.
By the construction of $\sigma_r$, $ts_1 > ts$. Now
since $w_1$ completed the third phase before $rd$ invoked, and
since by Observation \ref{obs:WRmon}, for each $bo$, $bo.ts$
is monotonically increasing,
when $rd$ invoked, for every set $S$ of $n-f$ base objects, the 
maximal $bo.ts$ of all $bo \in S$ is bigger than or equal to
$ts_1$, and thus bigger than $ts$. 
Therefore $rd$ set $\hat{ts}$, in the first phase, to timestamp
bigger than $ts$, and thus does not return $w$'s value. A
contradiction.
 
\end{proof}

\noindent The next corollary follows from Corollary
\ref{cor:WRspecification}, and Lemmas \ref{lem:WRorderbefore1},
\ref{lem:WRorderbefore2}, \ref{lem:WRorderbefore3}.

\begin{corollary}
\label{cor:WRregularity}

The algorithm simulates a strongly regular register.

\end{corollary}

\noindent The following theorem stems from Corollaries
\ref{cor:WRstorage}, \ref{cor:WRliveness}, and
\ref{cor:WRregularity}.

\begin{theorem}
\label{theorem:WR}

There is a FW-terminating algorithm that simulates a strongly
regular register, which storage is bounded by $(2f+k)2D$ bits,
and in runs with at most $c < k$ concurrent writes, the storage
is bounded by $(c+1)D/k$ bits. Moreover, in a run with a finite
number of writes, if all the writes are correct, the storage is
eventually reduced to $(2f+k)D/k$ bits.

\end{theorem}

\bibliographystyle{plain}
\bibliography{bibliography}
\end{document}